\documentclass[letterpaper, 10 pt, conference]{ieeeconf}
\IEEEoverridecommandlockouts
\usepackage{graphicx} 
\usepackage{cite}

\usepackage{amsmath}
\usepackage{amssymb}
\usepackage{hyperref}

\let\labelindent\relax
\usepackage{enumitem}

\newtheorem{theorem}{Theorem}
\newtheorem{proposition}{Proposition}
\newtheorem{lemma}{Lemma}

\newtheorem{definition}{Definition}
\newtheorem{remark}{Remark}

\newcommand{\R}{\mathbb{R}}
\newcommand{\N}{\mathbb{N}}

\newcommand{\sat}{\textbf{sat}}

\newcommand{\norm}[1]{\lVert #1 \rVert}

\newcommand{\ones}{\mathbf{1}}

\newcommand{\Ind}{\mathbb{I}}
\newcommand{\idx}{\mathcal{I}}
\newcommand{\bzeta}{\boldsymbol{\zeta}}
\newcommand{\bz}{\boldsymbol{z}}
\newcommand{\btau}{\boldsymbol{\tau}}
\newcommand{\bb}{\boldsymbol{b}}
\newcommand{\bmu}{\boldsymbol{\mu}}

\title{\LARGE \bf
A Continuous Time and State-Space Relaxation of the Linear Threshold Model with Nonlinear Opinion Dynamics
}
\author{Ian Xul Belaustegui$^{a}$, 
        Himani Sinhmar$^{a}$, 
        Ling-Wei Kong$^{b}$,
        Andrew Michael Hein$^{b}$,
        Naomi Ehrich Leonard$^{a}$
    \thanks{}
    \thanks{This research has been supported in part by
    the Air Force Office of Scientific Research grant FA9550-24-1-0002 
     and the William R.\ and Jane G.\ Schowalter Research Fund, and National Science Foundation grants EF-2222478 and IOS-2338596.}
    \thanks{$^{a}$Department of Mechanical and Aerospace Engineering, Princeton University; \texttt{\{ianxul,himani.sinhmar,naomi\}@princeton.edu}}%
    \thanks{$^{b}$Department of Computational Biology, Cornell University; \texttt{\{ling-wei.kong, amh433\}@cornell.edu}}%
}

\begin{document}

\maketitle

\begin{abstract}
    The Linear Threshold Model (LTM) is widely used to study the propagation of collective behaviors as complex contagions. However, its dependence on discrete states and timesteps restricts its ability to capture the multiple time-scales inherent in decision-making, as well as the effects of subthreshold signaling. To address these limitations, we introduce a continuous-time and state-space relaxation of the LTM based on the Nonlinear Opinion Dynamics (NOD) framework. By replacing the discontinuous step-function thresholds of the LTM with the smooth bifurcations of the NOD model, we map discrete cascade processes to the continuous flow of a dynamical system. We prove that, under appropriate parameter choices, activation in the discrete LTM guarantees activation in the continuous NOD relaxation for any given seed set. We establish explicit conditions for equivalence: by sufficiently bounding the social coupling parameter, the continuous NOD cascades exactly recover the cascades of the discrete LTM. We then illustrate how this NOD relaxation provides a richer analytical framework than the LTM, allowing for the exploration of cascades driven by strictly subthreshold inputs and the role of temporally distributed signals.
\end{abstract}

\section{Introduction}

The propagation of collective behaviors, from the adoption of social norms \cite{granovetter1978threshold-1e6, watts2002simple-6dc,centola2007complex-055} to coordinated escape responses in animal groups \cite{rosenthal2015revealing-d20,fahimipour2023wild} is frequently formalized as a complex contagion \cite{centola2007complex-055} and analyzed using the Linear Threshold Model (LTM) \cite{kempe2003maximizing-eb1, rossi2017threshold-568, acemoglu2011diffusion-032}. Although they are foundational, standard discrete-state threshold models possess fundamental limitations when applied to real-world biological and engineered systems. These models can explain many aspects of behavioral spread, but they cannot be used to explore questions of how cascades emerge in the first place. Nodes in a fully inactive system cannot transition to an active state, which excludes the possibility of spontaneous false alarms that are often observed in nature. Furthermore, their reliance on discrete timesteps abstracts away the intrinsically continuous timescales of individual decision-making processes \cite{hein2022ecological-43f}, preventing the study of how internal cognitive dynamics interact with external stimuli over time.

To address these limitations, we introduce a continuous-time and state-space framework based on Nonlinear Opinion Dynamics (NOD) \cite{bizyaeva2023nonlinear-758, leonard2023fast-129}. By considering the cascade process through the lens of nonlinear dynamical systems, we provide a mathematically rigorous bridge between discrete contagion models and continuous networked dynamics. We then show how the NOD framework, unlike the LTM, provides the means to explore emergence of cascades from subthreshold inputs and the role of temporally distributed inputs. 

Previous approaches have introduced continuous formulations of the LTM across different disciplines. Notably, the Generalized Linear Threshold model (GLT) introduces continuous stochastic waiting times for activation, yet the state space itself remains strictly binary \cite{ran2020generalized}. While allowing for richer temporal structures of cascades, this approach does not capture the continuous, dynamical evolution of agent opinions. Most closely related to our approach, the Continuous Threshold Model (CTM) replaces the discrete step function of the LTM with a smooth sigmoidal activation function, yielding a smooth flow \cite{zhong2019continuous-04c}. In the CTM, a positive sign of the continuous state corresponds to activation, and the exact LTM equivalence is recovered only in the limit as the slope of the sigmoid approaches infinity, forcing the continuous function back into a strict step-function. This is fundamentally different from our approach, where activation corresponds to an agent's opinion converging to one, and the equivalence of cascades with LTM is exactly recovered for an explicit range of parameter values. 

Given an instance of the LTM network dynamics, which are discrete-time and discrete-state, we use the term \textit{relaxation} to mean producing a continuous (almost everywhere differentiable) dynamical system that inherits the fundamental properties of the LTM---the notion of activation, the interaction network, and the activation thresholds---in such a way that similarity or equivalence of cascade behavior is a consequence of structural embedding. The NOD relaxation we propose does exactly this: activation becomes a convergence condition to an active stable equilibrium, the interaction matrix is inherited unchanged, and the step-function thresholds of the LTM become saddle-node bifurcations.

Our main contributions are as follows:

\begin{enumerate}[label=\roman*)]
    \item We establish a formal procedure for producing a Nonlinear Opinion Dynamics (NOD) relaxation of any given Linear Threshold Model (LTM). By replacing the discontinuous step-function thresholds of the LTM with the smooth bifurcations of the NOD model, we map the discrete cascade process to the continuous flow of a dynamical system.
    \item We prove that, under appropriate parameter choices, activation in the discrete LTM guarantees activation in the continuous NOD relaxation for any given seed set.
    \item We prove exact equivalence between the two models for an explicit range of parameters: by sufficiently bounding the social coupling parameter, the continuous NOD dynamics exactly recover the discrete cascade of the LTM for any given seed set.
    \item We show how the NOD relaxation expands on the LTM and, unlike the LTM, can be used to explore important facets of cascade processes: 1) the possibility of cascades from only subthreshold inputs and 2) the role of temporal separation of distributed inputs.
\end{enumerate}

Section~\ref{sec:notation} defines notation and Section~\ref{sec:LTM&FCM} provides a review of the LTM. In Section~\ref{sec:NOD} we introduce the NOD framework for continuous cascades and prove it is a relaxation of the LTM in Section~\ref{sec:RelaxLTM2NOD}. We illustrate its more expansive features in Section~\ref{sec:expand}. Section~\ref{sec:conclusions} provides conclusions.

\section{Notation}
\label{sec:notation}
\textbf{Important sets:} $\N$ is the natural numbers, $\R$ is the real numbers, and $\R_>:=(0,\infty)$, $\R_\geq:=[0,\infty)$. If $n\in\N\setminus\{0\}$, then $[n]:=\{1,\dots, n\}$. 
\textbf{Important vectors:} The all ones vector is $\ones:=\begin{pmatrix}1 &\dots &  1\end{pmatrix}^T$, with its dimension determined by the context.  
\textbf{Important functions and operations:} $\vee$ is the logical OR operator. $\Ind$ is the indicator function. $\odot$ is the Hadamard (entry-wise) product. $\sat:\R^n\to [0,1]^n$ is the piece-wise linear saturation function applied entry-wise, defined as $\sat(\boldsymbol{x})=\max(0, \min(1, \boldsymbol{x}))$. For a vector $\boldsymbol{x}\in\R^n$, $\norm{\boldsymbol{x}}_\infty:=\max\{\lvert x_i\rvert\mid i\in[n]\}$. For a matrix $A\in\R_\geq^{n\times n}$, $\norm{A}_\infty:=\norm{A\ones}_\infty$.

\section{The Linear Threshold Model}\label{sec:LTM&FCM}
We consider the deterministic (persistent) LTM as a discrete-time dynamical system defined in the following way \cite{kempe2003maximizing-eb1, watts2002simple-6dc}. Consider a directed graph $G=(V, E)$ with $n\in\N$ nodes, and non-negative edge weights given by a weighted adjacency matrix $A\in\R_\geq^{n\times n}$, with zero diagonal entries. We assume for simplicity that $V=[n]$. We also define a threshold vector $\btau=\begin{pmatrix}\tau_1 &\dots &  \tau_n\end{pmatrix}^T\in\R_>^n$, where each node has an associated threshold $\tau_i\in\R_>$. We denote the state of the network at time $t\in\N$ as $\bzeta(t)=\begin{pmatrix}\zeta_1(t) &\dots &  \zeta_n(t)\end{pmatrix}^T\in\{0,1\}^n$. Each node $i\in[n]$ has a binary state $\zeta_i(t)\in\{0,1\}$, and we say that node $i$ is active at time $t$ if $\zeta_i(t)=1$, and inactive if $\zeta_i(t)=0$. For each node, the time evolution of its state is determined by
\begin{equation}\label{eq:LTM-dyn}
    \zeta_i(t+1)=\zeta_i(t)\,\vee\, \Ind\left(((A\bzeta(t))_i> \tau_i)\right).
\end{equation}
That is, an agent will become active in the next timestep if and only if its state is already active, or the weighted sum of its active neighbors exceeds its threshold. Note that we adopt a strict inequality for the activation threshold, a slight variation from standard LTM formulations, which simplifies the correspondence with the relaxation later on. Notice also that the dynamics of the LTM are completely determined by the weighted adjacency matrix and the threshold vector, so we will refer to a tuple $(A,\btau)\in\R_\geq^{n\times n}\times\R_>^{n}$ where the diagonals of $A$ are zero, as an \textit{instance} of the LTM. 

A fundamental problem in this setting is characterizing the terminal state $\bzeta(\infty)$ as a function of the initial condition $\bzeta(0)$ \cite{kempe2003maximizing-eb1, rossi2017threshold-568}. That is, given a seed set of nodes $\mathcal{I}\subseteq V$, we set the initial states to be
$$
    \zeta_i(0)=\Ind(i\in\idx)
$$
and define the cascade (or spreading, or contagion) size as $\lim_{t\to\infty} \mathbf{1}^T\zeta(t)$. A full (or large) cascade (i.e., $\lim _{t\to\infty}\mathbf{1}^T\bzeta(t)\approx n$) is often of particular interest \cite{watts2002simple-6dc}. Notice that, since the state space is finite and the dynamics are monotone, it follows that the dynamics reach a steady state in at most $n$ steps. Thus, the cascade size is given by $\mathbf{1}^T\bzeta(n)$.

\begin{definition}\label{def:LTM-C}
    Given a LTM instance $(A, \btau)$, and a seed set $\mathcal{I}\subseteq [n]$, we define the LTM cascade set 
    \begin{equation}
        C_{LTM}(\mathcal{I}):=\{i\in[n]\mid\lim_{t\to\infty} \zeta_i(t)=1\},
    \end{equation}
    where state $\bzeta$ evolves according to \eqref{eq:LTM-dyn} with initial condition $\zeta_i(0)=\Ind(i\in\mathcal{I})$.
\end{definition}

A particular subclass of the LTM model is the fractional contagion model (FCM) \cite{watts2002simple-6dc}. In this case, given a fixed threshold $\tau\in(0,1]$, the individual thresholds for the LTM are chosen such that $\tau_i = \tau\,(A\mathbf{1})_i$. Given this, the FCM implements a decision rule where each node becomes active if and only if it was previously active, or if the fraction of social input that it observes is greater than $\tau$. Decision rules that appear to be well approximated by FCM have been observed in empirical studies of human and animal social systems (e.g., \cite{rosenthal2015revealing-d20,guilbeault2018complex} but see \cite{fahimipour2023wild}). 

For simplicity, throughout the following discussion we assume that the information graph $G$ of the LTM is strongly connected, meaning that there is a path from any node to any other. This is equivalent to the weighted adjacency matrix $A\in\R_\geq^{n\times n}$ being irreducible. This condition allows us to assume that $A$ has a simple real dominant eigenvalue $\lambda_{\max}>0$.

\begin{figure}
    \centering
    \vspace{2mm}
    \includegraphics[width=0.7\linewidth]{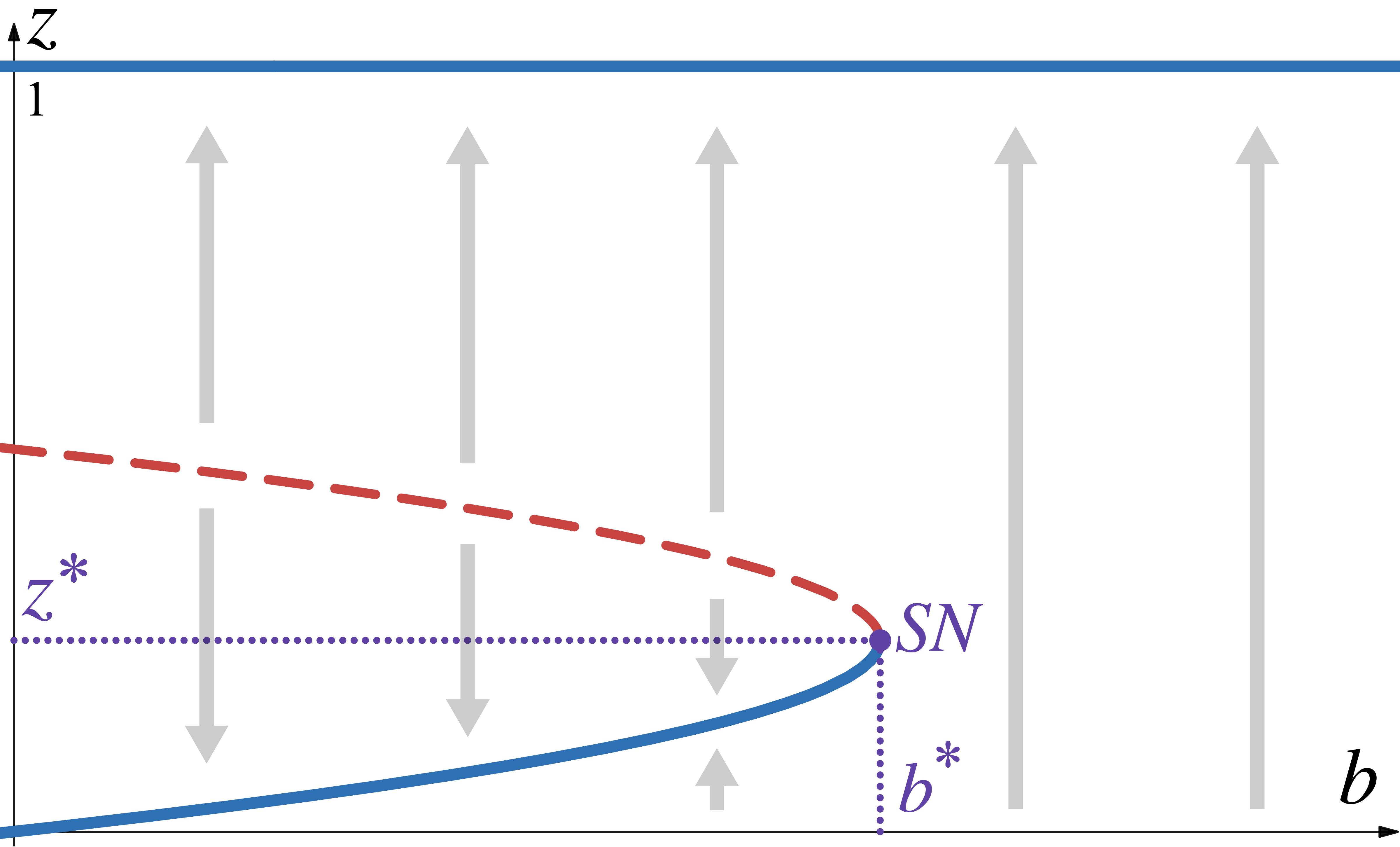}
    \caption{Bifurcation diagram of NOD dynamics \eqref{eq:NOD-dyn} for a single agent with respect to bifurcation parameter $b$. Thick solid blue and dashed red lines correspond to stable and unstable fixed points, respectively. The point marked $SN$ corresponds to the Saddle-Node bifurcation that marks the tipping point to activation. Values $b^*$ and $z^*$ are the activation threshold and the value achieved at the tipping point, respectively. Faint gray arrows represent the direction of the flow in each region. }
    \label{fig:b-bif}
\end{figure}

\section{The NOD Framework for Continuous Cascades}
\label{sec:NOD}
\subsection{Model Definition and Parameter Space}
We define an opinion dynamics model based on the nonlinear opinion dynamics (NOD) model \cite{bizyaeva2023nonlinear-758} in the following way. Consider $n$ agents, and a vector of opinion states $\bz\in[0,1]^n$. The opinions have continuous time dynamics
\begin{equation}\label{eq:NOD-dyn}
    \dot \bz=-\bz+\sat(\bmu\odot \bz+k(\bz\odot \bz) + \gamma A\bz + \bb),
\end{equation}
where $\bmu\in\R_>^n$ is a vector of heterogeneous gain on linear self-reinforcement, $k\in\R_\geq$ is a gain on nonlinear self-reinforcement, $\gamma\in\R_\geq$ is a rescaling of social interactions, $A$ is the matrix of social information weights of the LTM, and $\bb\in\R_\geq^n$ is a vector of inputs or biases. We denote a tuple $(A, \bmu, k, \gamma, \bb)\in\R_\geq^{n\times n}\times\R_>^{n}\times\R_\geq\times\R_\geq\times\R_\geq^{n}$ to be an instance of the NOD model. Notice that, in contrast with previous work, our NOD definition is not smooth everywhere, but is piecewise-smooth and continuous. However, the bifurcation we describe below always occurs away from the switching manifold, and in a region where the dynamics are locally smooth, so we still describe the NOD as smooth. Additionally, the saturation function can be approximated arbitrarily by smooth functions for which all the results in \cite{bizyaeva2023nonlinear-758, leonard2023fast-129} apply directly. 


In this paper we only consider heterogeneity in the linear gain $\mu_i$ across agents; the conditions on $\bmu$ and other parameters under which \eqref{eq:NOD-dyn} with $\bb=0$ corresponds to a given LTM instance are derived in Section~\ref{sec:RelaxLTM2NOD}. 

\subsection{Dynamical Properties of the NOD Model}\label{sec:ModelProps}
We discuss essential properties of the dynamics in \eqref{eq:NOD-dyn}, and characterize the threshold of each agent in terms of the model parameters. 
We have defined the opinions in \eqref{eq:NOD-dyn} to be in the unit hypercube $S:=[0,1]^n$. Indeed, the set $S$ is invariant for the dynamics. To see this, simply notice that for $i\in[n]$, if $z_i=1$ then $\sat(x)\leq 1$, $x\in \mathbb{R}$, implies $\dot z_i\leq 0$; and if $z_i=0$ then $\dot z_i\geq 0$, since all terms inside the saturation function are non-negative. Thus, the dynamics are well defined in $S$. 


\begin{remark}\label{rem:neutral-eq}
The neutral equilibrium $\bz=\boldsymbol{0}$ of the no-input system is asymptotically stable if $\norm{\bmu}_\infty<1-\gamma\lambda_{\max}$. 
\end{remark}
\begin{remark}\label{rem:monotone}
    For almost every initial state $\bz(0)\in S$, there exists a unique point $\bar \bz\in S$, such that $\bz(t)\to_{t\to\infty}\bar \bz$. 
\end{remark}
Remark \ref{rem:neutral-eq} follows immediately from ensuring the Jacobian of \eqref{eq:NOD-dyn} at the origin, $J=-I+\text{diag}(\bmu)+\gamma A$, is Hurwitz. This is consistent with the wider work on the indecision-breaking bifurcation in the NOD model \cite{bizyaeva2023nonlinear-758, leonard2023fast-129}. A tighter characterization of stability could be derived, but it is not required for our purposes. Remark \ref{rem:monotone} is a consequence of the system being a monotone cooperative dynamical system \cite{hirsch2006chapter-4b1}, with some minor consideration for the fact that $\sat$ is continuous but not everywhere differentiable. So limit cycles and more complex attractors need not be considered. 

We say that an agent or node is active if it converges to $1$ in forward time. The following proposition tells us that if an agent's opinion value starts close enough to a value of $1$, then it will stay active. 

\begin{proposition}\label{prop:ones-invariant}
    Let $(A,\bmu,k,\gamma,\bb)$ be an instance of NOD. If $k\geq1$, then for any $i\in[n]$, the subset $S_{z_i=1}:=\{\bz\in S\mid z_i=1\}$ is invariant. Furthermore, $S_{z_i=1}$ is asymptotically stable. 
\end{proposition}
\begin{proof}
    First, if $k\geq 1$, then $\mu_i+k>1$ by definition of NOD instance, and so $\sat(\mu_i z_i + kz_i^2 + {(\gamma A\bz)_i} + b_i)\rvert_{z_i=1}=1$, which implies $\left.{\dot z_i}\right\rvert_{z_i=1} = 0 $. So $S_{z_i=1}$ is invariant. Furthermore, by continuity there exists $\varepsilon\in(0,1)$ such that $\mu_i(1-\varepsilon)+k(1-\varepsilon)^2>1$. Thus, for every $\bz\in S$, if $\lvert 1-z_i \rvert < \varepsilon$, then $\dot z_i=-z_i+1$, which converges exponentially to $z_i=1$. 
\end{proof}

From the proof of Proposition \ref{prop:ones-invariant} it is clear that condition $k\geq1$ is not the tightest possible to guarantee our desired behavior, but for the theorems in Section \ref{sec:RelaxLTM2NOD} we will consider that $k\geq1$ for simplicity, as it does not affect the set of behaviors achievable. 

As an immediate consequence of Proposition \ref{prop:ones-invariant}, it follows that the full cascade ($\bz=\ones$) is an equilibrium (as it is the unique intersection $\cap_{i\in[n]}S_{z_i=1}=\{\ones\}$), and it is asymptotically stable. This suggests a new way of posing the question raised in Section \ref{sec:LTM&FCM} about understanding initial conditions that lead to full cascades: in the context of NOD it becomes that of understanding which points are in the basin of attraction of the stable equilibrium $\bz=\ones$. We do not explore this question here. We give the following definition in close correspondence with Definition~\ref{def:LTM-C}. 

\begin{definition}\label{def:NOD-C}
    Given a NOD instance $(A, \bmu, k,\gamma, \bb)$, and a seed set $\mathcal{I}\subseteq [n]$, we define the NOD cascade set 
    \begin{equation}
        C_{NOD}(\mathcal{I}):=\{i\in[n]\mid\lim_{t\to\infty} z_i(t)=1\},
    \end{equation}
    where state $\bz$ evolves according to \eqref{eq:NOD-dyn} with initial condition $z_i(0)=\Ind(i\in\mathcal{I})$. (Dependence of the cascade on the particular NOD instance is left implicit.) 
\end{definition}

We state the threshold condition for activation of an agent in the NOD model. 

\begin{proposition}\label{prop:b-threshold}
    Let $(A,\bmu,k,\gamma,\bb)$ be an instance of NOD, with $\norm{\bmu}_\infty<1$, and an initial condition $\bz(0)\in S$. For $i\in[n]$, denote $b_i^*:=\frac{(1-\mu_i)^2}{4k}$. If there exists $t^*\geq0$ and $\varepsilon>0$, such that $\forall t>t^*$, we have $\gamma(A\bz)_i+b_i> b_i^*+\varepsilon$, then $\lim_{t\to\infty} z_i(t)=1$. 
\end{proposition}
\begin{proof}
    Let $i\in[n]$ and $\bz\in S$ such that $\gamma(A\bz)_i+b_i>b_i^*+\varepsilon$ for some $\varepsilon>0$, and $z_i<1$. There are two cases. First, if the saturation function is active on the $i$th coordinate, then $\dot z_i=-z_i+1>0$. Otherwise, if the saturation is inactive,  
\begin{align*}
    \dot z_i&=-z_i+\mu_iz_i+kz_i^2+\gamma\cdot(A\bz)_i+b_i \\
                    &>kz_i^2+(\mu_i-1)z_i+b_i^* +\varepsilon  \\
                    &=k\left(z_i-\frac{(1-\mu_i)}{2k}\right)^2+\varepsilon\geq \varepsilon >0.
\end{align*} 
So, $\dot z_i\geq\min(-z_i+1, \varepsilon)>0$ for all $t>t^*$, hence  $\lim_{t\to\infty}z_i(t)=1$.
\end{proof}

Although in Proposition \ref{prop:b-threshold} we require the combined social and external input to remain above the threshold for all time, it would only need to be maintained long enough for the state to enter the basin of a fixed point in $S_{z_i=1}$. One way to understand where the threshold $b^*$ comes from is to regard it as a tipping point (saddle-node bifurcation) occurring with respect to parameter $b$ in \eqref{eq:NOD-dyn} when we consider a single uncoupled agent. This bifurcation is illustrated in Figure~\ref{fig:b-bif}. 

\begin{figure}
    \centering
    \vspace{2mm} 
    \includegraphics[width=1.0\linewidth]{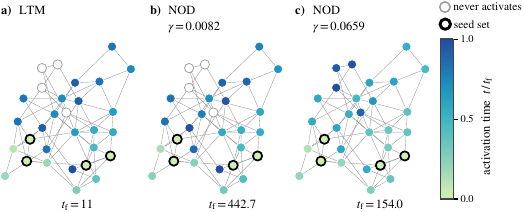}
    \caption{Simulated cascade in a random instance of LTM, and two NOD relaxations. Color bars show times of activation, where activation in NOD is taken to be when opinion crosses $0.95$ value. The seed set consists of $4$ nodes highlighted with black border.  Network and code can be found at https://github.com/ianxul/CDC2026-NOD-and-LTM.}
    \label{fig:cascade-examples}
\end{figure}

\section{An NOD-based relaxation of LTM}\label{sec:RelaxLTM2NOD}
We show how, given an instance of an LTM, it is possible to produce an instance of the NOD model defined above, such that this new model can be seen as a relaxation of the LTM model. By replacing the discontinuous step-function thresholds of the LTM with the smooth bifurcations of the NOD model, we map the discrete cascade process to the continuous flow of a dynamical system.

\subsection{Mapping LTM Thresholds to NOD Bifurcations}
Given an instance $(A,\btau)$ of the LTM as defined in Section~\ref{sec:LTM&FCM}, we show how to produce an instance of the NOD model that is a relaxation of the LTM to a continuous flow. In Proposition \ref{prop:b-threshold} we obtained an expression for the threshold ($b^*_i$) of each agent in terms of $k$ and $\mu_i$. If there is no exogenous input ($\bb=0$ in \eqref{eq:NOD-dyn}), then 
if after any time the social input is maintained strictly above an agent's threshold, 
that agent will become active. Thus, the condition for switching in the LTM model, $(A\bzeta)_i>\tau_i$, corresponds to the condition $(A\bz)_i>\frac{ b_i^*}{\gamma}=\frac{(1-\mu_i)^2}{4\gamma k}$. In order for the two models to match in behavior, we choose the $\mu_i$ in the NOD model so that $\tau_i=\frac{b_i^*}{\gamma}=\frac{(1-\mu_i)^2}{4\gamma k}$. This results in the expression $\mu_i=1-2\sqrt{\gamma k\tau_i}$, where we assume $\gamma$ and $k$ are chosen such that $\mu_i>0$ and also satisfy conditions described earlier ($k\geq 1$, $\gamma>0$). We can see that the constraints that we have found tell us that the valid values of the parameters require that $k$ is large enough and $\gamma$ is small enough. 

\begin{definition}
    \label{def:NODrelaxation}
    Given a LTM instance $(A,\btau)$, we say an NOD instance $(A,\bmu,k, \gamma, 0)$ is a \textbf{relaxation} of it if $k\geq 1$, $\gamma>0$ and $\forall i\in[n],$ $ \mu_i=1-2\sqrt{\gamma k\tau_i}>0$.
\end{definition}
This definition implies $\gamma k\norm{\btau}_\infty< 1/4$, for the NOD instance to be valid. We have the following result that relates the cascades in the LTM to those in its NOD relaxations. 

\begin{theorem}\label{thm:LTM2NOD}
    Let $(A,\btau)$ be an instance of the LTM and consider an NOD relaxation $(A,\bmu,k,\gamma,0)$. Then for any $ \mathcal{I}\subseteq[n],$ we have that $C_{LTM}(\mathcal{I})\subseteq C_{NOD}(\mathcal{I})$ .
\end{theorem}
\begin{proof}
    Denote the state $\bzeta(\ell)$ evolving according to the LTM dynamics \eqref{eq:LTM-dyn} in discrete time $\ell\in\N$.  Define, $C^{(\ell)}:=\{i\in[n]\mid \zeta_i(\ell)=1\}$. Then $C^{(0)}=\idx$ and $C^{(n)}=C_{LTM}(\idx)$. We show that for each $\ell\in\{0,...,n\}, $ $C^{(\ell)}\subseteq C_{NOD}(\idx)$, by induction on $\ell$. 
First, by definition, for all $i\in \idx$, $z_i(0)=1$, so by Proposition \ref{prop:ones-invariant} we conclude that $C^{(0)}=\idx\subseteq C_{NOD}(\idx)$. 

For the inductive step, take $\ell\in\{0,...,n-1\}$, and assume $C^{(\ell)}\subseteq C_{NOD}(\idx)$. Letting $j\in C^{(\ell+1)}$, we need to show that $j\in C_{NOD}(\idx)$. From the LTM dynamics, there are two cases: either $j\in C^{(\ell)}$ or  $\exists$$\delta\in\R_>$ such that $(A\bzeta(\ell))_j=\tau_j+\delta$. In the former case we are done by the inductive hypothesis. In the latter case, our inductive hypothesis $C^{(\ell)}\subseteq C_{NOD}(\idx)$ implies that $\forall i\in C^{(\ell)}$, $z_i\to_{t\to\infty} 1$, which in turn implies
$$
    \liminf_{t\to\infty} (A\bz(t))_j\geq (A\bzeta(\ell))_j=\tau_j+\delta=\frac{b_j^*}{\gamma}+\delta.
$$
Then $\liminf_{t\to\infty} \gamma(A\bz(t))_j\geq b_j^*+\gamma\delta$, so there exists $t^*>0$ such that for all $t\geq t^*$, $\gamma(A\bz(t))_j > (b_j^*+\gamma\delta)-\frac{\gamma\delta}{2}=b_j^*+\frac{\gamma\delta}{2}$. Using Proposition \ref{prop:b-threshold} 
$j\in C_{NOD}(\idx)$. Thus, $C^{(\ell+1)}\subseteq C_{NOD}(\idx)$.
By induction $C_{LTM}(\idx)=C^{(n)}\subseteq C_{NOD}(\idx)$.
\end{proof}

It is in this sense that we say that NOD is a relaxation of LTM, since the thresholds for each node match, meaning activation in LTM implies activation in NOD for the same seed set.  An NOD relaxation always exists, but is not unique, and its properties can be tuned by parameters $\gamma$ and $k$. An example of the preceding discussion is presented in Fig.~\ref{fig:cascade-examples}. This figure shows a cascade starting from $4$ seed nodes in the LTM (a) and in two NOD relaxations, one with small $\gamma$ (b) and one with large $\gamma$ (c). The cascade in the LTM is partial, and takes $11$ timesteps to complete. We can see that the cascade set in both NOD relaxations contains the LTM cascade set, and in fact it perfectly matches it in the case of small $\gamma$ (this is formalized in Thm. \ref{thm:NODlim2LTM}). It is interesting to note that, although the cascades in (a) and (b) match, in the LTM the nodes take on exactly $11$ colors, whereas in the NOD all $22$ nodes in the cascade take on different colors, corresponding to a continuous gradation of activation times, something that is absent in the LTM case.

The parameter $k$ modulates the strength of the nonlinear self-reinforcement. Higher values of $k$ produce steeper, more switch-like activation dynamics. 

From Remark \ref{rem:neutral-eq} it follows that a further condition, sufficient if we wish to impose stability of the neutral equilibrium $\bz=\boldsymbol{0}$  such that no perturbation of it leads to a cascade, is to require $\gamma<(1-\norm{\bmu}_\infty)/\lambda_{\max}$.

\begin{figure}
    \centering
    \vspace{2mm}
    \includegraphics[width=1.0\linewidth]{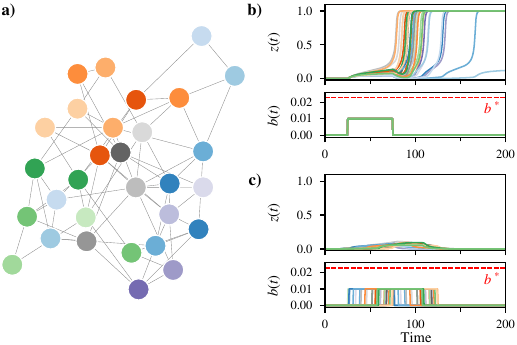}
    \caption{Cascades in the NOD model for distributed subthreshold inputs.  $a)$ Network structure, nodes are identified by color. This is the same network used in \ref{fig:cascade-examples}, but in this case the thresholds for each node are chosen to be equal $\tau_i=0.5$, leading to homogeneous $b_i^*$ in the NOD relaxation. $b)$ Top subplot shows opinion states ($\bz(t)$) of agents across time. Bottom subplot shows the applied time-varying input $b$ for each agent. The threshold for all agents is $b_i^*\approx 0.02$, shown as a dashed red line. In this case, a subthreshold input of $b_i=0.01$ is applied simultaneously to each agent as a step function for $50$ time units. $c)$ Similar to subplot (b). In this case, the input step function is applied with a delay selected uniformly randomly from the interval $[0,50].$} 
    \label{fig:dist-input}
\end{figure}

\subsection{Exact Equivalence and Qualitative Differences}
The fact that NOD dynamics evolve in continuous state and time means that the particular evolution of the cascade, while related as shown before to that of the LTM, evolves in a quite different way. For example, while any input above threshold will cause activation of the node, the speed with which this happens also depends on the magnitude of the input, with stronger inputs leading to faster responses. This means a seed set that is just slightly above the threshold for propagation can take considerably longer to cascade than a seed set producing an initial input well above the threshold of its neighbors, even when this difference is not noticeable in the LTM. In addition to this temporal difference, another consequence of the opinion state being continuous is that there is subthreshold signaling. That is, even when the input an agent receives is below its threshold, it will make its resting state positive, which in turn feeds back into the network. This fact implies that the relation in the previous section does not hold in the opposite direction, that is, we do \textit{not} have in general that $C_{NOD}(\mathcal{I})\subseteq C_{LTM}(\mathcal{I})$. This is illustrated in Fig. \ref{fig:cascade-examples}, where we can see that the cascade set of the NOD relaxation with high $\gamma$ leads to a full cascade, which is a strict superset of the LTM cascade set in A. However, as we show in the following, we can guarantee exact correspondence when $\gamma$ is made small enough. 

Intuitively, this follows from the fact that as the social coupling $\gamma$ decreases, the opinion value at the activation threshold (the value $z_i^*$ in Fig. \ref{fig:b-bif}) correspondingly decreases, and so the influence of subthreshold opinions on the network diminishes, driving the continuous cascade in NOD to match the discrete cascade in LTM. The following theorem shows that by sufficiently bounding the social coupling parameter, the continuous NOD relaxation recovers the exact discrete cascade of the LTM, independent of the initial seed set. First, we prove a lemma related to Proposition \ref{prop:b-threshold} that will be useful in proving Theorem \ref{thm:NODlim2LTM}, which tells us that if an agent starting at neutral opinion eventually becomes active, then it must have at some point received a combined social and external input greater than its threshold. 
\begin{lemma}\label{lem:must-cross-thresh}
    Let $(A,\bmu,k,\gamma,\bb)$ be an instance of NOD with initial condition $\bz(0)\in S$. Suppose $\norm{\bmu}_\infty<1$, and $k\in\R_>$ such that  $z_i^*:=\frac{1-\mu_i}{2k}<1$. Let $T\in\R_\geq$. If $z_i(0)=0$ and $\forall t\in[0,T]$, $\gamma(A\bz(t))_i+b_i < b_i^*$, then $\forall t\in[0,T]$, $z_i(t) < z_i^*$.
\end{lemma}
\begin{proof}
Notice that at $z_i=z_i^*$, 
\begin{align*}
    \phi_i&:=\mu_iz_i^*+kz_i^{*2}+\gamma(A\bz)_i+b_i \\
                                    &< \left(\frac{1-\mu_i}{2k}\right)\left(\mu_i+\frac{1-\mu_i}{2}\right) + \frac{(1-\mu_i)^2}{4k} \\
                                    &= \frac{(1-\mu_i)(1+\mu_i)}{4k} + \frac{(1-\mu_i)^2}{4k} \\
                                    &= \frac{1-\mu_i}{2k} = z_i^* < 1.
\end{align*}
So the saturation is inactive. From this we have that $\left.\dot z_i\right|_{z_i=z_i^*} = -z_i^*+\sat(\phi_i) = -z_i^*+\phi_i < -z_i^*+z_i^* = 0$. Now, suppose $\exists t\in[0,T]$, $z_i(t)\geq z_i^*$. Take $t^*:=\inf\{t\in[0,T]\mid z_i(t)\geq z_i^*\}>0$. Then, by continuity, $z_i(t^*)=z_i^*$ and $\forall t\in[0,t^*), z_i(t)<z_i^*$. But this would imply $\dot z_i(t^*)\geq 0$, which contradicts $\left.\dot z_i\right|_{z_i=z_i^*}<0$. 
\end{proof}

\begin{theorem}\label{thm:NODlim2LTM}
    Let $(A,\btau)$ be an instance of the LTM such that $\forall\bzeta\in\{0,1\}^n,\forall i\in[n], (A\bzeta)_i\neq\tau_i$. Then there exists a $\gamma^*\in\R_>$ such that for any NOD relaxation $(A,\bmu,k,\gamma,0)$ of the LTM instance with $\gamma\in(0,\gamma^*)$, it holds that $\forall\idx\subseteq[n], $ $C_{NOD}(\idx)=C_{LTM}(\idx)$.

\end{theorem}
\begin{proof}
    To establish a bound independent of the seed set, define $\delta_{gap}:=\min\{\tau_i-(A\bzeta)_i\mid \bzeta\in\{0,1\}^n, i\in [n], (A\bzeta)_i<\tau_i\}.$
    The minimum exists because it is over a finite set, and we have $\delta_{gap}>0$. Let $\idx\subseteq[n]$ be some seed set, and denote $\bar\bzeta(n):=\ones-\bzeta(n)$.
    Now, take 
    $$ \gamma^*:=\frac{\delta_{gap}^2}{\norm{\btau}_\infty\norm{A}_\infty^2}\leq\frac{k\delta_{gap}^2}{\norm{\btau}_\infty\norm{A}_\infty^2},
    $$
    and let $0<\gamma<\gamma^*$. Suppose $\Delta:=C_{NOD}(\idx)\setminus C_{LTM}(\idx)\neq\emptyset$. Let $i\in[n]\setminus C_{LTM}(\idx)$ be the index of the agent whose state $z_i$ first crosses its value $z_i^*=\frac{1-\mu_i}{2k}=\sqrt{\frac{\gamma\tau_i}{k}}$ at time $t^*$, and if there are multiple simultaneous crossings, take the minimum index. Note $i$ has to exist because $\Delta\neq\emptyset$. Then, we have for all $t\leq t^*$, 
    \begin{align*}\gamma(A\bz)_i+b_i&\leq\gamma(A\bzeta(n)+\norm{\bz^*}_\infty A\bar\bzeta(n))_i+0\\
                    &\leq \gamma(\tau_i-\delta_{gap}+\norm{\bz^*}_\infty \norm{A}_\infty)\\
                    &= b_i^*+\gamma\left(-\delta_{gap}+\norm{\bz^*}_\infty \norm{A}_\infty\right)\\
                    &=b_i^*+\gamma\left(-\delta_{gap}+ \sqrt{\frac{\gamma\norm{\btau}_\infty}{k}}\norm{A}_\infty\right)\\
                    &<b_i^*+\gamma\left(-\delta_{gap}+ \sqrt{\frac{k\delta_{gap}^2\norm{\btau}_\infty}{k\norm{\btau}_\infty\norm{A}_\infty^2}}\norm{A}_\infty\right)\\
                    &= b_i^*.
    \end{align*}
    So $\forall t\in[0,t^*], $ $\gamma(A\bz)_i<b_i^*$, which contradicts Lemma \ref{lem:must-cross-thresh}. Thus, $\Delta=\emptyset$, and so, together with Theorem \ref{thm:LTM2NOD}, we conclude $C_{NOD}(\idx)= C_{LTM}(\idx)$.
\end{proof}

The $\gamma^*$ defined above is certainly very conservative, and could be made tighter, but the key takeaway is that the choice of NOD relaxation can be tuned through the choice of parameters $\gamma$ and $k$ to match exactly the cascade sets of the corresponding LTM in almost all cases. The condition that all $(A\bzeta)_i\neq\tau_i$ is almost always satisfied, and when not, there is another equivalent instance of LTM a small perturbation away that does satisfy it. This is a non-degeneracy condition stating that it is never the case that some activation pattern produces a social input to a node that is exactly at its threshold. It is also worth noting that small $\gamma$ does not mean turning off the network. Both the threshold $\bb^*$ and the social input $\gamma(A\bz)$ scale linearly with $\gamma$, so the social activation condition is unchanged. What does change is the opinion value $\bz^*$ at the tipping point, which also vanishes as $\gamma\to0$. This means that when $\gamma$ becomes small, partially-activated agents contribute negligibly relative to fully activated ones. 



\section{The NOD relaxation expands on the LTM}
\label{sec:expand}
The behavior captured by the continuous dynamics of the NOD model is much richer than that of the LTM. In defining the NOD relaxation we have taken the bias or external information term $\bb$ to be zero. In the LTM, this term is not necessary because it would produce the same effect as reducing the thresholds. However, in the NOD model, a constant non-zero $\bb$ produces a persistent subthreshold signal. As a result it is possible to have input vectors $\bb$ that are subthreshold for all individual agents, but which produce a cascade at the group level. Cascade triggering by distributed inputs has been analyzed in NOD \cite{bizyaeva2021control-ee6, franci2021analysis-d82}; but our analysis extends this work by obtaining thresholds in closed form and inheriting their values from an underlying LTM. An example of this is presented in Fig. \ref{fig:dist-input}, where for an NOD relaxation of an LTM (subplot $a$) with constant threshold (leading to a constant threshold in the NOD) we apply an equal time-varying input $\bb(t)$ as a square wave with an amplitude less than half that of the threshold for a short interval, and this produces an almost full cascade (subplot $b$). This is possible due to subthreshold effects feeding back through the network. In fact, if $A$ is irreducible, then for any given non-zero input direction, the critical input in that direction at which a tipping point to a (partial) cascade occurs will be component-wise strictly less than the vector of isolated thresholds, since if any $b_i=b_i^*$, then the strictly positive social feedback would push the agent past its own tipping point (Prop. \ref{prop:b-threshold}).

We can also use the NOD model 
to consider time-varying inputs $\bb$, and how the temporal structure of the information received by the agents affects cascade outcomes. An example of this is shown in Fig. \ref{fig:dist-input}$c$, where the same square time-varying input is applied to all agents, but with a delay selected uniformly randomly to be no greater than the duration of the square wave. We can see that this change produces a subthreshold response in the agents, but it is not enough to drive the group beyond the collective tipping point to a cascade, even though there is a non-zero interval where all the inputs overlap. This behavior is not specific to the realization shown: across $100$ independent draws of the delays, with all other parameters held fixed, no agent activated in any run. This emphasizes the large role that the temporal structure of inputs can play in determining cascade outcomes, which are not captured in the LTM. 

Both 
aspects will be explored formally in future work. 

\section{Conclusion}
\label{sec:conclusions}
We have established a rigorous mathematical bridge between discrete contagion models and continuous networked dynamics by introducing a continuous-time and continuous-state relaxation of the Linear Threshold Model (LTM). By mapping the discontinuous step-function thresholds of the LTM directly to the smooth saddle-node bifurcations of the Nonlinear Opinion Dynamics (NOD) model, we retain previous results on cascade properties in the classical threshold model while embedding it within a rich, differentiable state-space. We formally proved that activation in the discrete LTM guarantees activation in the NOD relaxation for any given seed set, and established explicit parameter bounds on the social coupling under which the continuous flow exactly recovers the discrete cascade.

Beyond the exact equivalence regime, the true utility of this NOD relaxation lies in its ability to capture network phenomena that are not captured by discrete models. By allowing subthreshold signaling, the continuous framework shows how distributed, weak inputs can feed back through the network to trigger global cascades. Furthermore, it reveals the effects on collective decision-making of the temporal structure and delay of external signals. Future work will leverage this continuous formulation to systematically study the effects of time-varying inputs, structured network heterogeneity, and mixed feedback motifs on the emergence of collective behavior in both biological and engineered systems.

\bibliographystyle{IEEEtran}
\bibliography{refs}

\end{document}